\documentclass[11pt]{amsart}

\usepackage{fullpage}
\usepackage{amsmath, amssymb, enumerate, color, graphicx, bm, url}

\newtheorem{theorem}{Theorem}[section]

\newtheorem{corollary}[theorem]{Corollary}
\newtheorem{lemma}[theorem]{Lemma}

\newtheorem{definition}[theorem]{Definition}

\theoremstyle{remark}
\newtheorem{remark}{Remark}

\numberwithin{equation}{section}

\DeclareMathOperator*{\conv}{conv}
\DeclareMathOperator*{\E}{\mathbb{E}}

\DeclareMathOperator*{\Id}{Id}

\DeclareMathOperator*{\sign}{sign}
\DeclareMathOperator*{\supp}{supp}

\def \N {\mathbb{N}}
\def \R {\mathbb{R}}

\def \P {\mathbb{P}}

\def \NN {\mathcal{N}}

\def \MM {\mathcal{M}}

\def \e {\varepsilon}

\def \d {\delta}

\def \l {\lambda}

\def \< {\langle}
\def \> {\rangle}
\def \^ {\widehat}

\newcommand{\rv}[1]{#1}
\newcommand{\norm}[1]{\left \|#1\right \|}
\newcommand{\zeronorm}[1]{\norm{#1}_0}
\newcommand{\onenorm}[1]{\norm{#1}_1}
\newcommand{\twonorm}[1]{\norm{#1}_2}

\newcommand{\abs}[1]{\left | #1 \right |}
\renewcommand{\Pr}[1]{\P \left\{ #1 \rule{0mm}{3mm}\right\}}

\newcommand{\vect}[1]{\bm{#1}}
\newcommand{\mat}[1]{\bm{#1}}
\def \va {\vect{a}}

\def \vu {\vect{u}}
\def \vv {\vect{v}}
\def \vw {\vect{w}}
\def \vx {\vect{x}}
\def \vy {\vect{y}}
\def \vz {\vect{z}}

\def \mA {\mat{A}}

\title{One-bit compressed sensing by linear programming}

\author{Yaniv Plan}
\author{Roman Vershynin}

\date{September 19, 2011}

\address{Department of Mathematics,
  University of Michigan,
  530 Church St.,
  Ann Arbor, MI 48109, U.S.A.}

\email{\{yplan,romanv\}@umich.edu}

\subjclass[2000]{94A12; 60D05; 90C25}

\thanks{Y.P. is supported by an NSF Postdoctoral Research Fellowship under award No. 1103909. 
  R.V. is supported by NSF grants DMS 0918623 and 1001829.}

\begin{document}

\begin{abstract}
  We give the first computationally tractable and almost optimal solution to 
  the problem of one-bit compressed sensing, showing how to accurately recover an $s$-sparse vector $\vx \in \R^n$  
  from the signs of $O(s \log^2(n/s))$ random linear measurements of $\vx$.
  The recovery is achieved by a simple linear program.
  This result extends to approximately sparse vectors $\vx$. Our result is universal in the sense that 
  with high probability, one measurement scheme will successfully recover all sparse vectors simultaneously.  
  The argument is based on solving an equivalent geometric problem on random hyperplane tessellations. 
\end{abstract}

\maketitle

\section{Introduction}

Compressed sensing is a modern paradigm of data acquisition, which is having an impact on several disciplines, see \cite{Mackenzie2009}.
The scientist has access to a measurement vector $\vv \in \R^m$ obtained as
\begin{equation}							\label{v=Ax}
\vv = \mA \vx,
\end{equation}
where $\mA$ is a given $m \times n$ measurement matrix and $\vx \in \R^n$ is an unknown signal that one needs to recover from $\vv$.
One would like to take $m \ll n$, rendering $\mA$ non-invertible; 
the key ingredient to successful recovery of $\vx$ is take into account its assumed structure -- sparsity.  
Thus one assumes that $\vx$ has at most $s$ nonzero entries, although the support pattern is unknown.  
The strongest known results are for random measurement matrices $\mA$.  
In particular, if $\mA$ has Gaussian i.i.d. entries, then we may take $m = O(s \log(n/s))$ and still recover 
$\vx$ exactly with high probability \cite{Candes2006, Candes2005}; see \cite{Vershynin2010} for an overview.
Furthermore, this recovery may be achieved in polynomial time by solving the convex minimization program
\begin{equation}						\label{eq:cs recovery}
\min \onenorm{\vx'} \quad \text{subject to} \quad \mA \vx' = \vv.
\end{equation}
Stability results are also available when noise is added to the problem \cite{Candes2006a, Candes2007, Bickel2009, Wojtaszczyk2009}.

However, while the focus of compressed sensing is signal recovery with minimal information, 
the classical set-up \eqref{v=Ax}, \eqref{eq:cs recovery} assumes {\em infinite bit precision} of the measurements.  
This disaccord raises an important question: how many bits per measurement (i.e. per coordinate of $\vv$) are sufficient
for tractable and accurate sparse recovery? This paper shows that {\em one bit per measurement} is enough. 

There are many applications where such severe quantization may be inherent or preferred ---
analog-to-digital conversion \cite{Laska2010, Jacques2011}, binomial regression in statistical modeling 
and threshold group testing \cite{Damaschke2005}, to name a few.


\subsection{Main results}

This paper demonstrates that a simple modification of the convex program \eqref{eq:cs recovery} 
is able to accurately estimate $\vx$ from extremely quantized measurement vector
$$
\vy = \sign(\mA \vx).
$$
Here $\vy$ is the vector of signs of the coordinates of $\mA \vx$.\footnote{To be precise, for a scalar $z \neq 0$ we define $\sign(z) = z/\abs{z}$, and $\sign(0) = 0$.  We allow the $\sign$ function to act on a vector by acting individually on each element.}

Note that $\vy$ contains no information about the magnitude of $\vx$, 
and thus we can only hope to recover the normalized vector $\vx/\|\vx\|_2$.
This problem was introduced and first studied by Boufounos and Baraniuk \cite{Boufounos2008}
under the name of {\em one-bit compressed sensing}; some related work is summarized in Section~\ref{sec:prior work}.

We shall show that the signal can be accurately recovered by solving the following convex 
minimization program
\begin{equation}		\label{eq:convex program}
\min \onenorm{\vx'} \quad \text{subject to} \quad \sign(\mA \vx') \equiv \vy \quad {\text{and}} 
\quad \onenorm{\mA \vx'} = m.
\end{equation}
The first constraint, $\sign(\mA \vx') \equiv \vy$, keeps the solution consistent with the measurements.  It is defined by the relation $\< \va_i, \vx' \> \cdot \vy_i \geq 0$ for $i = 1, 2, \hdots, m$, where $\va_i$ is the $i$-th row of $\mA$.
The second constraint, $\onenorm{\mA \vx'} = m$, serves to prevent the program from returning a zero solution.  
Moreover, this constraint is linear as it can be represented as one linear equation
$\sum_{i=1}^m y_i \< \va_i, \vx'\> = m$ where $y_i$ denote the coordinates of $\vy$.
Therefore \eqref{eq:convex program} is indeed a convex minimization program; furthermore one can easily  
represent it as a linear program, see \eqref{eq:linear-program} below. 
Note also that the number $m$ in \eqref{eq:convex program} is chosen for convenience of the analysis; 
it can be replaced by any other fixed positive number.

\begin{theorem}[Recovery from one-bit measurements]			\label{thm:main}
  Let $n, m, s > 0$, and let $\mA$ be an $m \times n$ random matrix with independent standard normal entries.
  Set 
  \begin{equation}							\label{eq:delta}
  \delta = C \left(\frac{s}{m} \log(2n/s) \log(2n/m + 2m/n)\right)^{1/5}.
  \end{equation}
  Then, with probability at least $1 - C \exp(- c \delta m)$, 
  the following holds uniformly for all signals $\vx \in \R^n$ satisfying 
  $\onenorm{\vx}/\twonorm{\vx} \leq \sqrt{s}$. 
  Let $\vy = \sign(\mA\vx)$. 
  Then the solution $\hat{\vx}$ of the convex minimization program \eqref{eq:convex program}
  satisfies  
  $$
  \twonorm{\frac{\hat{\vx}}{\twonorm{\hat{\vx}}} - \frac{\vx}{\twonorm{\vx}}} \leq \delta.
  $$
\end{theorem}

Here and thereafter $C$ and $c$ denote positive absolute constants; other standard notation is explained 
in Section~\ref{sec:notation}.

\begin{remark}[Effective sparsity]
  The Cauchy-Schwarz inequality implies that $\onenorm{\vx}/\twonorm{\vx} \leq \sqrt{\zeronorm{\vx}}$
  where $\|\vx\|_0 = |\supp(\vx)|$ is the number of nonzero elements of $\vx$. 
  Therefore one can view 
  the parameter $(\onenorm{\vx}/\twonorm{\vx})^2$ as a measure of {\em effective sparsity} of the signal $\vx$. 
  The effective sparsity is thus a real valued and robust extension of the sparsity parameter $\|\vx\|_0$, 
  which allows one to handle approximately sparse vectors. 
\end{remark}

Let us then state the partial case of Theorem~\ref{thm:main} for sparse signals: 

\begin{corollary}[Sparse recovery from one-bit measurements]			\label{cor:sparse}
  Let $n, m, s > 0$, and set $\d$ as in \eqref{eq:delta}.
  Then, with probability at least $1 - C \exp(- c \delta m)$, 
  the following holds uniformly for all signals $\vx \in \R^n$ satisfying $\|x\|_0 \le s$. 
  Let $\vy = \sign(\mA\vx)$. 
  Then the solution $\hat{\vx}$ of the convex minimization program \eqref{eq:convex program}
  satisfies  
  $$
  \twonorm{\frac{\hat{\vx}}{\twonorm{\hat{\vx}}} - \frac{\vx}{\twonorm{\vx}}} \leq \delta.  
  $$
\end{corollary}

\begin{remark}[Number of measurements]
  The conclusion of Corollary~\ref{cor:sparse} can be stated
  in the following useful way. \rv{With high probability, 
  {\em an arbitrarily accurate estimation of every $s$-sparse vector $\vx$ can be achieved from 
  $$
  m = O(s \log^2(n/s)) 
  $$
  one-bit random measurements.} 
  The implicit factor in the $O(\cdot)$ notation depends only on the desired accuracy level $\d$; 
  more precisely $m \sim \d^{-5} s \log^2(n/s)$ up to an absolute constant factor.}
  The same holds if $\vx$ is only effectively $s$-sparse as in Theorem~\ref{thm:main}.  
  The central point here is that the number of measurements is almost linear in the sparsity $s$, 
  which can be much smaller than the ambient dimension $n$.
\end{remark}

\begin{remark}[Non-gaussian measurements]
  Most results in compressed sensing, and in random matrix theory in general, 
  are valid not only for Gaussian random matrices but also for general random 
  matrix ensembles. 
  In one-bit compressed sensing, since the measurements $\sign(\mA \vx)$ do not depend on the 
  scaling of the rows of $\mA$, it is clear that our results will not change if the rows of $\mA$ 
  are sampled independently from {\em any rotationally invariant distribution} in $\R^n$ (for example, 
  the uniform distribution on the unit Euclidean sphere $S^{n-1}$). 
  
  However, in contrast to the widespread universality phenomenon, 
  one-bit compressed sensing cannot be generalized to some of the simplest {\em discrete distributions},
  such as Bernoulli. Indeed, suppose the entries of $\mA$ are independent 
  $\pm 1$ valued symmetric random variables. Then for the vectors 
  $\vx = (1,0,0,\ldots,0)$ and $\vx' = (1,\frac{1}{2},0,\ldots,0)$
  one can easily check that $\sign(\mA \vx) = \sign(\mA \vx')$ for any number 
  of measurements $m$. So one-bit measurements can not distinguish between 
  two fixed distinct signals $\vx$ and $\vx'$ no matter how many measurements are taken.
\end{remark}

\begin{remark}[Optimality]
  For a fixed level of accuracy, our estimate on the number of measurements $m = O(s \log^2(n/s))$ 
  matches the best known number of measurements in the classical (not quantized) compressed sensing problem 
  up to the exponent $2$ of the logarithm, and up to an absolute constant factor. 
  However, we believe that the exponent $2$ can be reduced to $1$. 
  We also believe that the error $\delta$ in Theorem~\ref{thm:main} 
  may decrease more quickly as $s/m \rightarrow 0$. 
  In particular, Jacques et al. \cite{Jacques2011} demonstrate that $\vx$ is exactly sparse and is estimated using an $\ell_0$-minimization-based approach, 
  the error is upper bounded as $\delta = O((s/m)^{1 - o(1)} \log n)$; 
  they also demonstrate a lower error bound $\d  = \Omega(s/m)$ regardless of what algorithm is used.  In fact, such a result is not possible when $\vx$ is only known to be effectively sparse (i.e., $\onenorm{\vx}/{\twonorm{\vx}} \leq \sqrt{s}$).  Instead, the best possible bound is of the form $\delta = O(\sqrt{(s/m) \log(n/s)})$ (this can be checked via entropy arguments).  We believe this is achievable (and is optimal) 
  for the convex program \eqref{eq:convex program}. 
\end{remark}

\subsection{Prior work}				\label{sec:prior work}
While there have been several numerical results for quantized compressed sensing 
\cite{Boufounos2008, Boufounos2009, Boufounos2010, Laska2010, Zymnis2010}, 
as well as guarantees on the convergence of many of the algorithms used for these numerical results, 
theoretical accuracy guarantees have been much less developed.  
One may endeavor to circumvent this problem by considering quantization errors as a source of noise, 
thereby reducing the quantized compressed sensing problem to the noisy classical compressed sensing problem.  
Further, in some cases the theory and algorithms of noisy compressed sensing may be adapted to this problem as in \cite{Zymnis2010, Dai2009, jacques2011dequantizing, sun2009}; the method of quantization may be specialized in order to minimize the recovery error.  
As noted in \cite{Laska2011} if the range of the signal is unspecified, then such a noise source is unbounded, 
and so the classical theory does not apply. 
However, in the setup of our paper we may assume without loss of generality that $\twonorm{\vx} = 1$, and thus 
it is possible that the methods of Candes and Tao \cite{Candes2007} can be adapted to derive
a version of Corollary~\ref{cor:sparse} for a fixed sparse signal $\vx$. 
Nevertheless, we do not see any way to deduce by these methods a uniform result over all sparse signals $\vx$.

In a complementary line of research Ardestanizadeh et al. \cite{ardestanizadeh2009}  consider compressed sensing with a finite number of bits per measurement. 
However, the number of bits per measurement there is not one (or constant);  
this number depends on the sparsity level $s$ and the dynamic range of the signal $\vx$.
Similarly, in the work of Gunturk et al. \cite{gunturk2010, gunturk2010paper} 
on sigma-delta quantization, the number of bits per measurement 
depends on the dynamic range of $\vx$.  On the other hand, by considering sigma-delta quantization and multiple bits, the Gunturk et al. are able to provide excellent guarantees on the speed of decay of the error $\delta$ as $s/m$ decreases.

The framework of one-bit compressed sensing was introduced by Boufounos and Baraniuk in \cite{Boufounos2008}.
Jacques et al. \cite{Jacques2011} show that $O(s \log n)$ one-bit measurements 
are sufficient to recover an $s$-sparse vector with arbitrary precision; their results are also robust to bit flips.  In particular, their results require the estimate $\hat{\vx}$ to be as sparse as $\vx$, have unit norm, and be consistent with the data.  The difficulty is that the first two of these constraints are non-convex, and thus the only known program which is known to return such an estimate is $\ell_0$ minimization with the unit norm constraint---this is generally considered to be intractable.
Gupta et al. \cite{Gupta2010} demonstrate that one may tractably recover the support of $\vx$ from $O(s \log n)$ measurements.  
They give two measurement schemes.  One is non-adaptive, but the number of measurements has a quadratic dependence on the dynamic range of the signal.  The other has no such dependence but is adaptive. 
Our results settle several of these issues: (a) we make no assumption about the dynamic range of the signal, 
(b) the one-bit measurements are non-adaptive, and 
(c) the signal is recovered by a tractable algorithm (linear programming).

\subsection{Notation and organization of the paper}		\label{sec:notation}

Throughout the paper, $C$, $c$, $C_1$, etc.~denote absolute constants whose values may change from line to line. 
For integer $n$, we denote $[n] = \{1,\ldots,n\}$.  Vectors are written in bold italics, e.g., $\vx$, and their coordinates are written in plain text so that the $i$-th component of $\vx$ is $x_i$.  
For a subset $T \subset [n]$, $\vx_T$ is the vector $\vx$ restricted to the elements indexed by $T$.  
The $\ell_1$ and $\ell_2$ norms of a vector $\vx \in \R^n$ are defined as
$\|\vx\|_1 = \sum_{i=1}^n |x_i|$ and $\|\vx\|_2 = (\sum_{i=1}^n x_i^2)^{1/2}$ respectively.
The number of non-zero coordinates of $\vx$ is denoted by $\|\vx\|_0 = |\supp(\vx)|$.
The unit balls with respect to $\ell_1$ and $\ell_2$ norms are denoted by 
$B_1^n = \{ x \in \R^n:\; \|x\|_1 \le 1 \}$ and $B_2^n = \{ x \in \R^n:\; \|x\|_2 \le 1 \}$ respectively.
The unit Euclidean sphere is denoted $S^{n-1} = \{ x \in \R^n:\; \|x\|_2 = 1 \}$.

The rest of the paper is devoted to proving Theorem~\ref{thm:main}. In Section~\ref{sec:strategy}
we reduce this task to the following two ingredients: (a) Theorem~\ref{thm:effective sparsity} which states
states that {\em a solution to \eqref{eq:convex program} is effectively sparse}, 
and (b) Theorem~\ref{thm:non-convex} which analyzes a simpler but non-convex version of \eqref{eq:convex program}
where the constraint $\|\mA\vx'\|_1 = m$ is replaced by $\|\vx'\|_2 = 1$. 
The latter result can be interpreted in a geometric way in terms of {\em random hyperplane tessellations}
of a subset $K$ of the Euclidean sphere, specifically for the set of effectively sparse signals $K = S^{n-1} \cap \sqrt{s} B_1^n$.
In Section~\ref{sec:geometry} we estimate the metric entropy of $K$, and we use this
in Section~\ref{sec:tessellations} to prove our main geometric result of independent interest: {\em $m = O(s \log (n/s))$ 
random hyperplanes are enough to cut $K$ into small pieces}, 
yielding that all cells of the resulting tessellation have arbitrarily small diameter. 
This will complete part (b) above. 
For part (a), we prove Theorem~\ref{thm:effective sparsity} on the effective sparsity of solutions
in Section~\ref{sec:effective sparsity}. The proof is based on counting all possible solutions of 
\eqref{eq:convex program}, which are the vertices of the feasible polytope. 
This will allow us to use standard concentration inequalities from the Appendix
and to conclude the argument by a union bound.

\subsection*{Acknowledgement}

\rv{The authors are grateful to Sinan~G\"unt\"urk for pointing out an inaccuracy in the statement
of Lemma~\ref{lem:entropy compressible} in an earlier version of this paper. }

\section{Strategy of the proof}				\label{sec:strategy}

Our proof of Theorem~\ref{thm:main} has two main ingredients which we explain in this section. Throughout the 
paper, $\va_i$ will denote the rows of $\mA$, which are i.i.d. standard normal vectors in $\R^n$.

Let us revisit the second constraint $\onenorm{\mA \vx'} = m$ in the convex minimization program \eqref{eq:convex program}.
Consider a fixed signal $\vx'$ for the moment. Taking the expectation with respect to the random matrix $\mA$, we see that 
$$
\E \|\mA\vx'\|_1 = \sum_{i=1}^m \E |\< \va_i, \vx'\> | = c m \|\vx'\|_2
$$
where $c = \sqrt{2/\pi}$. Here we used that the first absolute moment of the standard normal random variable equals $c$.
So {\em in expectation}, the constraint $\onenorm{\mA \vx'} = m$ is equivalent to $\|\vx'\|_2 = 1$ up to constant factor $c$. 

This observation suggests that we may first try to analyze the simpler minimization program
\begin{equation}							\label{eq:non-convex program}
\min \onenorm{\vx'} \quad \text{subject to} \quad \sign(\mA \vx') = \vy \quad {\text{and}} 
\quad \twonorm{\vx'} = 1.
\end{equation}
This optimization program was first proposed in \cite{Boufounos2008}.
Unfortunately, it is non-convex due to the constraint $\twonorm{\vx'} = 1$, and therefore seems to be computationally intractable.
On the other hand, we find that the non-convex program \eqref{eq:non-convex program} is more amenable to theoretical analysis
than the convex program \eqref{eq:convex program}.

\medskip

The first ingredient of our theory will be to demonstrate that the non-convex optimization program \eqref{eq:non-convex program}
leads to accurate recovery of an effectively sparse signal $\vx$. One can reformulate this as a geometric problem
about {\em random hyperplane tessellations}. We will discuss tessellations in Section~\ref{sec:tessellations}; 
the main result of that section is Theorem~\ref{thm:separation} which immediately implies the following result:

\begin{theorem}		\label{thm:general-noiseless}
  Let $n,m,s>0$, and set
  \begin{equation}							\label{eq:delta non-convex}
  \delta = C \left(\frac{s}{m} \log(2n/s)\right)^{1/5}.
  \end{equation}
  Then, with probability at least $1 - C \exp(- c \delta m)$, the following holds uniformly for 
  all $\vx, \hat{\vx} \in \R^n$ that satisfy $\|\vx\|_2 = \|\hat{\vx}\|_2 = 1$, $\|\vx\|_1 \le \sqrt{s}$, $\|\hat{\vx}\|_1 \le \sqrt{s}$:
  $$
  \sign(\mA\hat{\vx}) = \sign(\mA\vx) \quad \text{implies} \quad \twonorm{\hat{\vx} - \vx} \leq \delta.
  $$
\end{theorem}

Theorem~\ref{thm:general-noiseless} yields a version of our main Theorem~\ref{thm:main} 
for the non-convex program \eqref{eq:non-convex program}:

\begin{theorem}[Non-convex recovery]			\label{thm:non-convex}
  Let $n, m, s > 0$, and set $\d$ as in \eqref{eq:non-convex program}.
  Then, with probability at least $1 - C \exp(- c \delta m)$, 
  the following holds uniformly for all signals $\vx \in \R^n$ satisfying 
  $\onenorm{\vx}/\twonorm{\vx} \leq \sqrt{s}$. 
  Let $\vy = \sign(\mA\vx)$. 
  Then the solution $\hat{\vx}$ of the non-convex minimization program \eqref{eq:non-convex program}
  satisfies  
  $$
  \twonorm{\hat{\vx} - \frac{\vx}{\twonorm{\vx}}} \leq \delta.
  $$
\end{theorem}

\begin{proof}
We can assume without loss of generality that $\twonorm{\vx} =1$ and thus $\onenorm{\vx} \le \sqrt{s}$. 
Since $\vx$ is feasible for the program \eqref{eq:non-convex program}, we also have $\onenorm{\hat{\vx}} \leq \onenorm{\vx} \le \sqrt{s}$,
and thus $\hat{\vx} \in S^{n-1}$. Therefore Theorem \ref{thm:general-noiseless} applies to $\vx, \hat{\vx}$, and it yields
that $\twonorm{\hat{\vx} - \vx} \leq \delta$ as required.
\end{proof}

\begin{remark}[Prior work] \label{rem:prior work}
  A version of Theorem~\ref{thm:general-noiseless} was recently proved in \cite{Jacques2011} for 
  exactly sparse signals $\vx, \hat{\vx}$, i.e.~such that $\|\vx\|_2 = \|\hat{\vx}\|_2 = 1$, 
  $\|\vx\|_0 \le s$, $\|\hat{\vx}\|_0 \le s$. This latter result holds with $\d = C (s/m)^{1-o(1)} \log (2n)$.
  However, from the proof of Theorem~\ref{thm:non-convex} given above one sees that the result
  of \cite{Jacques2011} would not be sufficient to deduce our main results, even Corollary~\ref{cor:sparse} 
  for exactly sparse vectors. The reason is that our goal is to solve a tractable program that involves the $\ell_1$ norm, and thus we cannot directly assume that our estimate will be in the low-dimensional set of exactly sparse vectors.
   Our proof of Theorem~\ref{thm:general-noiseless} has to overcome 
  some additional difficulties compared to \cite{Jacques2011} caused by the absence of any control of the 
  supports of the signals $\vx, \hat{\vx}$.  In particular, the metric entropy of the set of unit-normed, sparse vectors only grows logarithmically with the inverse of the covering accuracy.  This allows the consideration of a very fine cover in the proofs in \cite{Jacques2011}.  In contrast, the metric entropy of the set of vectors satisfying $\twonorm{\vx} \leq 1$ and $\onenorm{\vx} \leq \sqrt{s}$ is much larger at fine scales, thus necessitating a different strategy of proof.  
\end{remark}

\medskip

Theorem~\ref{thm:main} would follow if we could demonstrate that 
the convex program \eqref{eq:convex program} and the non-convex program \eqref{eq:non-convex program}
were equivalent. Rather than doing this explicitly, we shall prove that the solution $\hat{\vx}$
of the convex program \eqref{eq:convex program}
essentially preserves the effective sparsity of a signal $\vx$, 
and we finish off by applying Theorem~\ref{thm:general-noiseless}.

\begin{theorem}[Preserving effective sparsity]			\label{thm:effective sparsity}
  Let $n,s>0$ and suppose that $m \geq C s \log(n/s)$. 
  Then, with probability at least $1 - C \exp(- c m)$, the following holds uniformly for all signals $\vx$ 
  satisfying $\onenorm{\vx}/\twonorm{\vx} \leq \sqrt{s}$. 
  Let $\vy = \sign(\mA\vx)$. 
  Then the solution $\hat{\vx}$ of the convex minimization program \eqref{eq:convex program} satisfies
  $$
  \frac{\onenorm{\hat{\vx}}}{\twonorm{\hat{\vx}}} 
    \leq \frac{\onenorm{\vx}}{\twonorm{\vx}} \cdot C \sqrt{\log(2n/m + 2m/n)}.
  $$
\end{theorem}

This result is the second main ingredient of our argument, and it will be proved in Section~\ref{sec:effective sparsity}.
Now we are ready to deduce Theorem~\ref{thm:main}.

\begin{proof}[Proof of Theorem~\ref{thm:main}.]
Consider a signal $\vx$ as in Theorem~\ref{thm:main}, so $\onenorm{\vx}/\twonorm{\vx} \leq \sqrt{s}$. 
In view of the application of Theorem~\ref{thm:effective sparsity}, we may assume without loss of generality
that $m \geq C s \log(n/s)$. Indeed, otherwise we have $\d \ge 2$ and the conclusion of Theorem~\ref{thm:main}
is trivial. So Theorem~\ref{thm:effective sparsity} applies and gives
$$
\frac{\onenorm{\hat{\vx}}}{\twonorm{\hat{\vx}}} \le C \sqrt{s \log(2n/m + 2m/n)} =: \sqrt{s_0}.
$$
Also, as we noted above, $\onenorm{\vx}/\twonorm{\vx} \leq \sqrt{s} \le \sqrt{s_0}$.
So Theorem~\ref{thm:general-noiseless} applies for the normalized vectors 
$\vx/\|\vx\|_2$, $\hat{\vx}/\|\hat{\vx}\|_2$ and for $s_0$. Note that $\sign(\mA\hat{\vx}) = \sign(\mA\vx) = \vy$
because $\hat{\vx}$ is a feasible vector for the program \eqref{eq:convex program}. Therefore
Theorem~\ref{thm:general-noiseless} yields
$$
\twonorm{\frac{\hat{\vx}}{\twonorm{\hat{\vx}}} - \frac{\vx}{\twonorm{\vx}}} \leq \delta
$$
where
$$
\d = C \left(\frac{s_0}{m} \log(2n/s)\right)^{1/5}
= C' \left(\frac{s}{m} \log(2n/s) \log(2n/m + 2m/n)\right)^{1/5}.
$$
This completes the proof.
\end{proof}

For the rest of the paper, our task will be to prove the two ingredients above --
Theorem~\ref{thm:general-noiseless}, which we shall relate to a more general 
hyperplane tessellation problem, and Theorem~\ref{thm:effective sparsity} on 
the effective sparsity of the solution.

\section{Geometry of signal sets}				\label{sec:geometry}

Our arguments are based on the geometry of the set of effectively $s$-sparse signals
$$
K_{n,s} := \big\{ \vx \in \R^n :\; \|\vx\|_2 \le 1, \; \|\vx\|_1 \le \sqrt{s} \big\}
= B_2^n \cap \sqrt{s} B_1^n
$$
and the set of $s$-sparse signals 
$$
S_{n,s} := \big\{ \vx \in \R^n :\; \|\vx\|_2 \le 1, \; \|\vx\|_0 \le s \big\}.
$$
While the set $S_{n,s}$ is not convex, $K_{n,s}$ is, and moreover
it is a convexification of $S_{n,s}$ in the following sense.  Below, for a set $K$, we define $\conv(K)$ to be its convex hull.

\begin{lemma}[Convexification]		\label{lem:characterization}
  One has 
  $\conv(S_{n,s}) \subset K_{n,s} \subset 2\conv(S_{n,s})$.
\end{lemma}

\begin{proof}
The first containment follows by Cauchy-Schwartz inequality, which implies for 
each $\vx \in S_{n,s}$ that $\onenorm{\vx} \leq \sqrt{s}$.
The second containment is proved using a common technique in the compressed sensing literature.  
Let $\vx \in K_{n,s}$.  Partition the support of $\vx$ into disjoint subsets $T_1, T_2, \hdots$ 
so that $T_1$ indexes the largest $s$ elements of $\vx$ (in magnitude), $T_2$ indexes the next $s$ largest elements, and so on.  
Since all $\vx_{T_i} \in S_{n,s}$, in order to complete the proof it suffices to show that 
\[
\sum_{i \ge 1} \twonorm{\vx_{T_i}} \leq 2.
\]
To prove this, first note that $\twonorm{\vx_{T_1}} \leq \twonorm{\vx} \leq 1$.  
Second, note that for $i \geq 2$, each element of $\vx_{T_i}$ is bounded in magnitude by $\onenorm{\vx_{T_{i-1}}}/s$, 
and thus $\twonorm{\vx_{T_i}} \leq \onenorm{\vx_{T_{i-1}}}/\sqrt{s}$. Combining these two facts we obtain
\begin{equation}
\label{eq:convex-containment}
\sum_{i \ge 1} \twonorm{\vx_{T_i}} 
\leq 1 + \sum_{i \geq 2} \twonorm{\vx_{T_i}} 
\leq 1 + \sum_{i \ge 2} \onenorm{\vx_{T_i}}/\sqrt{s} 
\leq 1 + \onenorm{\vx}/\sqrt{s} \leq 2,
\end{equation}
where in the last inequality we used that $\|\vx\|_1 \le \sqrt{s}$ for $\vx \in K_{n,s}$. The proof is complete.
\end{proof}

\medskip

Our arguments will rely on entropy bounds for the set $K_{n,s}$. 
Consider a more general situation, where $K$ is a bounded subset of $\R^n$ and $\e>0$ is a fixed number.
A subset $\NN \subseteq K$ is called an {\em $\e$-net} of $K$ if for every $\vx \in K$
one can find $\vy \in \NN$ so that $\|\vx-\vy\|_2 \le \e$. The minimal cardinality of an $\e$-net 
of $K$ is called the {\em covering number} and denoted $N(K, \e)$.
The number $\log N(K,\e)$ is called the {\em metric entropy} of $K$. 
The covering numbers are (almost) increasing by inclusion: 
\begin{equation}							\label{eq:monotonicity}
K' \subseteq K \quad \text{implies} \quad N(K',2\e) \le N(K,\e).
\end{equation}

Specializing to our sets of signals $K_{n,s}$ and $S_{n,s}$, we 
come across a useful example of an $\e$-net:

\begin{lemma}[Sparse net]				\label{lem:sparse net}
  Let $s \le t$. Then $S_{n,t} \cap K_{n,s}$ is an $\sqrt{s/t}$-net of $K_{n,s}$.
\end{lemma}

\begin{proof}
Let $\vx \in K_{n,s}$, and let $T \subseteq [n]$ denote the set of the indices of the $t$ largest coefficients 
of $\vx$ (in magnitude). Using the decomposition $\vx = \vx_T + \vx_{T^c}$ and noting that 
$\vx_T \in S_{n,t} \cap K_{n,s}$, we see that it suffices to check that $\|\vx_{T^c}\|_2 \le \sqrt{s/t}$.
This will follow from the same steps as in \eqref{eq:convex-containment}.  In particular, we have
\[\twonorm{\vx_{T^c}} \leq \onenorm{\vx}/{\sqrt{\abs{T}}} \leq \sqrt{s/t}\]
as required.
\end{proof}

Next we pass to quantitative entropy estimates. \rv{}
The entropy of the Euclidean ball can be estimated using a standard volume comparison argument, 
as follows (see \cite[Lemma 4.16]{Pisier}):
\begin{equation}							\label{eq:entropy ball}
N(B_2^n,\e) \le (3/\e)^n, \quad \e \in (0,1).
\end{equation}
From this we deduce a known bound on the entropy of $S_{n,s}$: 

\begin{lemma}[Entropy of $S_{n,s}$]			\label{lem:entropy sparse}
  For $\e \in (0,1)$ \rv{and $s \le n$}, we have
  $$
  \log N(S_{n,s}, \e) \le s \log \Big( \frac{9 n}{\e s} \Big).
  $$
\end{lemma}

\begin{proof}
We represent $S_{n,s}$ as the union of the unit Euclidean balls $B_2^n \cap \R^I$ in all $s$-dimensional coordinate subspaces, 
$I \subset [n], \, |I|=s$.
Each ball $B_2^n \cap \R^I$ has an $\e$-net for  of cardinality at most $(3/\e)^s$, according to \eqref{eq:entropy ball}.
The union of these nets forms an $\e$-net of $S_{n,s}$, and since the number of possible $I$ is 
\rv{$\binom{n}{\lfloor s \rfloor}$, 
the resulting net has cardinality at most 
$\binom{n}{\lfloor s \rfloor} (3/\e)^{\lfloor s \rfloor} \le (3en/\e s)^s$.}
Taking the logarithm completes the proof. 
\end{proof}

As a consequence, we obtain an entropy bound for $K_{n,s}$:

\begin{lemma}[Entropy of $K_{n,s}$]		\label{lem:entropy compressible}
  For $\e \in (0,1)$, we have
  \rv{
  \begin{align*}
  \log N(K_{n,s}, \e) 
  &\le \begin{cases}
    n \log\big(\frac{6}{\e}\big) & \text{if }\, 0 < \e <2 \sqrt{\frac{s}{n}} \\
    \frac{4 s}{\e^2} \, \log \big( \frac{9\e n}{s} \big) & \text{if }\, 2\sqrt{\frac{s}{n}} \leq \e \leq 1
  \end{cases} \\
  & \leq \frac{C s}{\e^2} \, \log \Big( \frac{2 n}{s} \Big).
  \end{align*}
  }
\end{lemma} 

\begin{proof}
\rv{
First note that $K_{n,s} \subset B_2^n$. Then the monotonicity property \eqref{eq:monotonicity} followed by
the volumetric estimate \eqref{eq:entropy ball} yield the first desired bound
$N(K_{n,s}, \e) \leq N(B_2^n, \e/2) \leq n \log(6/\e)$ for all $\e \in (0,1)$.
}

\rv{Next, suppose that $2 \sqrt{\frac{s}{n}} < \e <1$.  Then set $t := 4s/\e^2 \le n$. 
Lemma~\ref{lem:sparse net} states that $S_{n,t} \cap K_{n,s}$ is an $(\e/2)$-net of $K_{n,s}$. 
Furthermore, to find an $(\e/2)$-net of $S_{n,t}$, we use Lemma~\ref{lem:entropy sparse} for $\e/4$ and for $t$.}
Taking into account the monotonicity property \eqref{eq:monotonicity}, we see that 
there exists an $(\e/2)$-net $\NN$ of $S_{n,t} \cap K_{n,s}$ and such that 
$$
\log |\NN| \le t \log \Big( \frac{36 n}{\e t} \Big)
= \frac{4 s}{\e^2} \, \log \Big( \frac{9\e n}{s} \Big).
$$
It follows that $\NN$ is an $\e$-net of $K_{n,s}$, and its cardinality is as required. 
\end{proof}

\section{Random hyperplane tessellations}		\label{sec:tessellations}

In this section we prove a generalization of Theorem \ref{thm:general-noiseless}.
We consider a set $K \subseteq \R^n$ and a collection of $m$ random hyperplanes in $\R^n$, 
chosen independently and uniformly from the Haar measure. 
The resulting partition of $K$ by this collection of hyperplanes is 
called a {\em random tessellation} of $K$. 
The cells of the tessellation 
are formed by intersection of $K$ and the $m$ random half-spaces with particular orientations. 
The main interest in the theory of random tessellations is the typical shape of the cells. 

\begin{figure}[htp]		
  \centering \includegraphics[height=2.7cm]{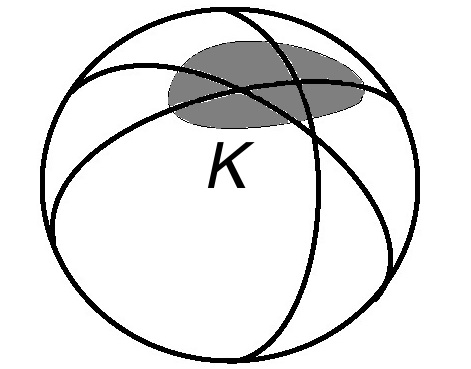} 
  \caption{Hyperplane tessellation of a subset $K$ of a sphere}
  \label{fig:tessellation}
\end{figure}

We shall study the situation where $K$ is a subset of the sphere $S^{n-1}$, see Figure~\ref{fig:tessellation}.
The particular example of $K = S^{n-1}$ is a natural model of random hyperplane tessellation in the sperical space $S^{n-1}$. 
The more classical and well studied model of random hyperplane tessellation is in Euclidean space $\R^n$, 
where the hyperlanes are allowed to be affine, see \cite{MS-tessellations} for the history of this field. 
The random hyperplane tessellations of the sphere is studied in particular in \cite{Miles-tessellations}.

Here we focus on the following question.
{\em How many random hyperplanes ensure that all the cells of the tessellation of $K$ have small diameter}
(such as $1/2$)?
For the purposes of this paper, we shall address this problem for a specific set, namely for 
$$
K = S^{n-1} \cap \sqrt{s} B_1^n = S^{n-1} \cap K_{n,s}.
$$
We shall prove that $m = O(s \log(n/s))$ hyperplanes suffice with high probability.
Our argument can be extended to more general sets $K$, but we defer generalizations to a later paper. 

\begin{theorem}[Random hyperplane tessellations]					\label{thm:tessellation}
  Let $s \le n$ and $m$ be positive integers. Consider the tessellation of the set 
  $K = S^{n-1} \cap \sqrt{s} B_1^n$ by $m$ random hyperplanes in $\R^n$
  chosen independently and uniformly from the Haar measure. 
  Let $\d \in (0,1)$, and assume that 
  $$
  m \ge C \d^{-5} s \log(2n/s). 
  $$
  Then, with probability at least $1-2\exp(-\d m)$, 
  all cells of the tessellation of $K$ have diameter at most $\d$.
\end{theorem}

It is convenient to represent the random hyperplanes in Theorem~\ref{thm:tessellation} as 
$(\va_i)^\perp$, $i=1,\ldots,m$, where $\va_i$ are i.i.d. standard normal vectors in $\R^n$.
The claim that all cells of the tessellation of $K$ have diameter at most $\d$ can be restated in the following way. 
Every pair of points $\vx, \vy \in K$ satisfying $\|\vx-\vy\|_2 > \d$ is separated by at least one of the hyperplanes, 
so there exists $i \in [m]$ such that 
$$
\< \va_i, \vx\> > 0, \quad \< \va_i, \vy \> < 0.
$$
Theorem~\ref{thm:tessellation} is then a direct consequence of the following slightly stronger result. 

\begin{theorem}[Separation by a set of hyperplanes]					\label{thm:separation}
  Let $s \le n$ and $m$ be positive integers. Consider the set $K = S^{n-1} \cap \sqrt{s} B_1^n$ 
  and independent random vectors $\va_1, \ldots, \va_m \sim \NN(0, \Id)$ in $\R^n$.
  Let $\d \in (0,1)$, and assume that 
  $$
  m \ge C \d^{-5} s \log(2n/s). 
  $$
  Then, with probability at least $1-2\exp(-\d m)$, the following holds. 
  For every pair of points $\vx,\vy \in K$ satisfying $\|\vx-\vy\|_2 > \d$, 
  there is a set of at least $c \d m$ of the indices $i \in [m]$ that satisfy 
  $$
  \< \va_i, \vx \> > c\d, \quad \< \va_i, \vy \> < -c\d.
  $$
\end{theorem}

We will prove Theorem~\ref{thm:separation} by the following covering argument, 
which will allow us to uniformly handle all pairs $\vx, \vy \in K$ satisfying $\|\vx-\vy\|_2 > \d$. 
We choose an $\e$-net $\NN_\e$ of $K$ as in Lemma~\ref{lem:entropy compressible}.
We decompose the vector $\vx = \vx_0 + \vx'$ where $\vx_0 \in \NN_\e$ is a ``center''
and $\vx' \in \e B_2^n \cap K$ is a ``tail'', and we do similarly for $\vy$.
An elementary probabilistic argument and a union bound will allow us to nicely separate each pair of 
centers $\vx_0, \vy_0 \in \NN_\e$ satisfying $\|\vx_0-\vy_0\|_2 > \d$ by $\Omega(m)$ hyperplanes. 
(Specifically, it will follow that $\< \va_i, \vx_0 \> > c\d$, $\< \va_i, \vy_0 \> < -c\d$ for at least $c\d m$ of the indices $i \in [m]$.)

Furthermore, the tails $\vx', \vy' \in \e B_2^n \cap \sqrt{s} B_1^n$ can be uniformly controlled using Lemma~\ref{lem:L1 RIP},
which implies that all tails are in a good position with respect to $m-o(m)$ hyperplanes. 
(Specifically, for small $\e$ one can deduce that $|\< \va_i, \vx' \> | < c\d/2$, $|\< \va_i, \vy' \> | < c\d/2$ 
for at least $m-c\d m/2$ of the indices $i \in [m]$.)
Putting the centers and the tails together, we shall conclude that $\vx$ and $\vy$ are 
separated at least $\Omega(m) + m-o(m) > \Omega(m)$ hyperplanes, as required.

\medskip

Now we present the full proof of Theorem~\ref{thm:separation}. 

\subsection{Step 1: Decomposition into centers and tails}	\label{s:decomposition}

Let $\e \in (0,1)$ be a number to be determined later. 
Let $\NN_\e$ be an $\e$-net of $K$. Since $K \subseteq K_{n,s}$, Lemma~\ref{lem:entropy compressible}
along with monotonicity property of entropy \eqref{eq:monotonicity} guarantee that $\NN_\e$ can be chosen so that 
\begin{equation}							\label{eq:net size}
\log |\NN_\e| \le \frac{C s}{\e^2} \, \log \Big( \frac{2 n}{s} \Big).
\end{equation}

\begin{lemma}[Decomposition into centers and tails]		\label{lem:decomposition}
  Let $t = 4s/\e^2$. Then every vector $\vx \in K$ can be represented as 
  \begin{equation}							\label{eq:decomposition}
  \vx = \vx_0 + \e \vx'
  \end{equation}
  where $\vx_0 \in \NN_\e$, $\vx' \in K_{n,t}$.
\end{lemma}

\begin{proof}
Since $\NN_\e$ is an $\e$-net of $K$, representation \eqref{eq:decomposition} holds for some
$\vx' \in B_2^n$. Since $K_{n,t} = B_2^n \cap \sqrt{t} B_1^n$, it remains to check that $\vx' \in \sqrt{t} B_1^n$.
Note that $\vx \in K \subset \sqrt{s}B_1^n$ and $\vx_0 \in \NN_\e \subset K \subset \sqrt{s}B_1^n$. By the triangle 
inequality this implies that $\e \vx' = \vx-\vx_0 \in 2\sqrt{s}B_1^n$.
Thus $\vx' \in (2\sqrt{s}/\e)B_1^n = \sqrt{t}B_1^n$, as claimed. 
\end{proof}

\subsection{Step 2: Separation of the centers}

Our next task is to separate the centers $\vx_0$, $\vy_0$ of each pair of points $\vx,\vy \in K$ 
that are far apart by $\Omega(m)$ hyperplanes. 
For a fixed pair of points and for one hyperplane, it is easy to estimate the probability of a nice separation.

\begin{lemma}[Separation by one hyperplane]  	\label{lem:separation hyperplane}
  Let $\vx, \vy \in S^{n-1}$ and assume that $\|\vx-\vy\|_2 \ge \d$ for some $\d>0$.
  Let $\va \sim \NN(0, \Id)$. Then for $\d_0 = \d/12$ we have
  \[
  \Pr{\< \va, \vx\> > \d_0, \; \< \va, \vy \> < -\d_0} \geq \d_0.
  \]
\end{lemma}

\begin{proof}
Note that
\begin{align*}
  \Pr{\< \va, \vx\> > \d_0, \; \< \va, \vy \> < -\d_0} &= \Pr{\< \va, \vx\> > 0 \;\mathrm{ and }\; \< \va, \vy\> < 0 \;\mathrm{ and }\; \< \va,\vx\> \notin (0, \d_0] \;\mathrm{ and }\;  \< \va, \vy \> \notin [-\d_0, 0)}\\
& \geq 1 - \Pr{\< \va, \vx\> \leq 0 \:\mathrm{or} \; \< \va, \vy\> \geq 0} - \Pr{\< \va,\vx\> \in (0, \d_0]}- \Pr{\< \va, \vy \> \in [-\d_0, 0)}.
\end{align*}
The inequality above follows by the union bound.  Now, since $\< \va, \vx \> \sim \NN(0,1)$ we have
\[
\Pr{\< \va, \vx \> \in (0, \d_0]} \leq \frac{\delta_0}{\sqrt{2 \pi}}
\quad \text{and} \quad
\Pr{\< \va, \vy \> \in [-\d_0, 0)} \leq \frac{\delta_0}{\sqrt{2 \pi}}.
\]
Also, denoting the geodesic distance in $S^{n-1}$ by $d(\cdot, \cdot)$ it is not hard to show that 
\[
\Pr{\< \va, \vx\> \leq 0 \:\mathrm{or} \; \< \va, \vy\> \geq 0} = 1 - \frac{d(\vx,\vy)}{2\pi}
\le 1 - \frac{\|\vx-\vy\|_2}{2\pi} \le 1 - \frac{\d}{2\pi}
\]
(see \cite[Lemma 3.2]{Goemans1995}).  Thus
\[
\Pr{\< \va, \vx\> > \d_0, \; \< \va, \vy \> < -\d_0} 
\ge \frac{\d}{2\pi} - \frac{2 \delta_0}{\sqrt{2 \pi}} \geq \d_0
\]
as claimed.
\end{proof}

Now we will pay attention to the number of hyperplanes that nicely separate a given pair of points. 

\begin{definition}[Separating set]
  Let $\d_0 \in (0,1)$. The separating index set of a pair of points $\vx,\vy \in S^{n-1}$ is defined as 
  $$
  I_{\d_0}(\vx,\vy) := \left\{ i \in [m]: \; \< \va_i, \vx\> > \d_0, \; \< \va_i, \vy \> < -\d_0 \right\}.
  $$
\end{definition}

The cardinality $|I_{\d_0}(\vx,\vy)|$ is a binomial random variable, which is the sum of $m$ indicator
functions of the independent events $\{\< \va_i, \vx\> > \d_0, \; \< \va_i, \vy \> < -\d_0\}$. The probability of each such event 
can be estimated using Lemma~\ref{lem:separation hyperplane}.
Indeed, suppose $\|\vx-\vy\|_2 \ge \d$ for some $\d>0$, and let $\d_0 = \d/12$.
Then the probability of each of the events above is at least $\d_0$. 
Then $|I_{\d_0}(\vx,\vy)| \sim \text{Binomial}(m,p)$ with $p > \d_0$. 
A standard deviation inequality (e.g. \cite[Theorem~A.1.13]{Alon-Spencer}) yields
\begin{equation}							\label{eq:separating set deviation}
\Pr{ |I_{\d_0}(\vx,\vy)| < \d_0 m/2 } \le e^{-\d_0 m/8}.
\end{equation}

Now we take a union bound over pairs of centers in the net $\NN_\e$ that was chosen in the beginning 
of Section~\ref{s:decomposition}. 

\begin{lemma}[Separation of the centers]			\label{lem:separation centers}
  Let $\e, \d \in (0,1)$, and let $\NN_\e$ be an $\e$-net of $K$ whose cardinality satisfies \eqref{eq:net size}.
  Assume that 
  \begin{equation}							\label{eq:centers m}
  m \ge \frac{C_1s}{\e^2\d} \log \Big( \frac{2 n}{s} \Big).
  \end{equation}  
  Then, with probability at least $1-\exp(-\d m/100)$, the following event holds:
  \begin{equation}							\label{eq:separation centers}
  \text{For every } \vx_0, \vy_0 \in \NN_\e \text{ such that } \|\vx_0-\vy_0\|_2 > \d, 
  \text{ one has } |I_{\d/12}(\vx_0,\vy_0)| \ge \d m/24.
  \end{equation}
\end{lemma}

\begin{proof}
For a fixed pair $\vx_0, \vy_0$ as above, we can rewrite \eqref{eq:separating set deviation} as 
$$
\Pr{ |I_{\d/12}(\vx_0,\vy_0)| < \d m/24 } \le e^{-\d m/96}.
$$
A union bound over all pairs $\vx_0, \vy_0$ implies that the event in \eqref{eq:separation centers}
fails with probability at most 
$$
|\NN_\e|^2 \cdot e^{-\d m/96}.
$$
By \eqref{eq:net size} and \eqref{eq:centers m}, this quantity is further bounded by
$$
\exp \Big[ \frac{C s}{\e^2} \, \log \Big( \frac{2 n}{s} \Big) - \frac{\d m}{96} \Big]
\le \exp(-\d m/100)
$$
provided the absolute constant $C_1$ is chosen sufficiently large. The proof is complete.
\end{proof}

\subsection{Step 3: Control of the tails}

Now we provide a uniform control of the tails $\vx' \in K_{n,t}$ that arise from the decomposition given in
Lemma~\ref{lem:decomposition}. 
The next result is a direct consequence of Lemma~\ref{lem:L1 RIP}.

\begin{lemma}[Control of the tails]			\label{lem:tails}
  Let $t \ge 1$ and let $\va_1,\ldots,\va_m \sim \NN(0,\Id)$ be independent random vectors in $\R^n$.
  Assume that 
  \begin{equation}							\label{eq:tails m}
  m \ge Ct \log(2n/t).
  \end{equation}  
  Then, with probability at least $1 - 2 \exp(-cm)$, the following event holds: 
  $$
  \sup_{\vx' \in K_{n,t}} \frac{1}{m} \sum_{i=1}^m \abs{\< \va_i, \vx'\>} \le 1.
  $$
\end{lemma}

\subsection{Step 4: Putting the centers and tails together}

Let $\e = c_0 \d^2$ for a sufficiently small absolute constant $c_0>0$. 
To control the tails, we choose an $\e$-net $\NN_\e$ of $K$ as in Lemma~\ref{lem:separation centers},
and we shall apply this lemma with $\d/2$ instead of $\d$. 
Note that requirement \eqref{eq:centers m} becomes 
$$
m \ge C_2 \d^{-5} s \log \Big( \frac{2 n}{s} \Big),
$$
and it is satisfied by the assumption of Theorem~\ref{thm:separation}, for a sufficiently large absolute constant $C$.
So Lemma~\ref{lem:separation centers} yields that with probability at least $1-\exp(-\d m/200)$, 
the following separation of centers holds:
\begin{equation}							\label{eq:centers}
\text{For every } \vx_0, \vy_0 \in \NN_\e \text{ such that } \|\vx_0-\vy_0\|_2 > \d/2, 
\text{ one has } |I_{\d/24}(\vx_0,\vy_0)| \ge \d m/48.
\end{equation}

To control the tails, we choose $t = 4s/\e^2 \sim s/\d^4$ as in Decomposition Lemma~\ref{lem:decomposition},
and we shall apply Lemma~\ref{lem:tails}. Note that requirement \eqref{eq:tails m} becomes
$$
m \ge C_3 \d^{-4} s \log \Big( \frac{C_3 n}{s} \Big),
$$
and it is satisfied by the assumption of Theorem~\ref{thm:separation}, for a sufficiently large absolute constant $C$.
So Lemma~\ref{lem:tails} yields that with probability at least $1 - 2 \exp(-cm)$, the following control of tails holds: 
\begin{equation}							\label{eq:tails}
\text{For every } \vx' \in K_{n,t}, \text{ one has } \frac{1}{m} \sum_{i=1}^m \abs{\< \va_i, \vx'\>} \le 1.
\end{equation}

Now we combine the centers and the tails. With probability at least $1 - 2 \exp(-c \d m)$, both events 
\eqref{eq:centers} and \eqref{eq:tails} hold. Suppose both these events indeed hold, and consider a pair of vectors
$\vx,\vy \in K$ as in the assumption, so $\|\vx-\vy\|_2 > \d$. We decompose these vectors according to 
Lemma~\ref{lem:decomposition}: 
\begin{equation}							\label{eq:decomposition x y}
\vx = \vx_0 + \e x', \quad \vy = \vy_0 + \e \vy'
\end{equation}
where $\vx_0, \vy_0 \in \NN_\e$ and $\vx', \vy' \in K_{n,t}$. 
By the triangle inequality and the choice of $\e$, the centers are far apart: 
$$
\|\vx_0 - \vy_0\|_2 \ge \|\vx - \vy\|_2 - 2\e 
> \d - 2\e = \d - 2c_0\d^2 \ge \d/2. 
$$
Then event \eqref{eq:centers} implies that the separating set 
\begin{equation}							\label{eq:I}
I_0:= I_{\d/24}(\vx,\vy) \text{ satisfies } |I_0| \ge \d m/48.
\end{equation}

Furthermore, using \eqref{eq:tails} for the tails $\vx'$ and $\vy'$ we see that 
$$
\frac{1}{m} \sum_{i=1}^m \abs{\< \va_i, \vx'\>} + \frac{1}{m} \sum_{i=1}^m \abs{\< \va_i, \vy'\>}
\le 2. 
$$
By Markov's inequality, the set 
$$
I' := \left\{ i \in [m]: \; \abs{\< \va_i, \vx'\>} + \abs{\< \va_i, \vy'\>} \le \frac{192}{\d} \right\}
\text{ satisfies } |(I')^c| \le \frac{\d m}{96}.
$$

We claim that 
$$
I: = I_0 \cap I'
$$
is a set of indices $i$ that satisfies the conclusion of Theorem~\ref{thm:separation}. 
Indeed, the number of indices in $I$ is as required since 
$$
|I| \ge |I_0| - |(I')^c| \ge \frac{\d m}{48} - \frac{\d m}{96} = \frac{\d m}{96}.
$$
Further, let us fix $i \in I$. Using decomposition \eqref{eq:decomposition x y} we can write
$$
\< \va_i, \vx\> = \< \va_i, \vx_0\> + \e \< \va_i, \vx'\> .
$$
Since $i \in I \subseteq I_0 = I_{\d/24}(\vx,\vy)$, we have $\< \va_i, \vx_0\> > \d/24$, while 
from $i \in I \subseteq I'$ we obtain $\< \va_i, \vx'\> \ge -192/\d$. Thus
$$
\< \va_i, \vx\> > \frac{\d}{24} - \frac{192\e}{\d} \ge \frac{\d}{30},
$$
where the last estimate follows by the choice of $\e = c_0 \d^2$ for a sufficiently small absolute 
constant $c_0>0$.
In a similar way one can show that 
$$
\< \va_i, \vy\> < -\frac{\d}{24} + \frac{192\e}{\d} \le -\frac{\d}{30}.
$$
This completes the proof of Theorem~\ref{thm:separation}.
\qed

\section{Effective sparsity of solutions}		\label{sec:effective sparsity}

In this section we prove Theorem~\ref{thm:effective sparsity} about the effective sparsity 
of the solution of the convex optimization problem \eqref{eq:convex program}.
Our proof consists of two steps -- a lower bound for $\|\hat{\vx}\|_2$ proved in Lemma~\ref{lem:min-l2} below, 
and an upper bound on $\|\vx\|_2$ which we can deduce from Lemma~\ref{lem:L1 RIP} in the Appendix. 

\begin{lemma}[Euclidean norm of solutions]				\label{lem:min-l2}
  Let $n, m > 0$.
  Then, with probability at least $1 - C \exp(-c m \log (2n/m + 2m/n))$, 
  the following holds uniformly for all signals $\vx \in \R^n$. 
  Let $\vy = \sign(\mA\vx)$. 
  Then the solution $\hat{\vx}$ of the convex minimization program \eqref{eq:convex program}
  satisfies  
  $$
  \twonorm{\hat{\vx}} \geq c/\sqrt{\log(2n/m + 2m/n)}.
  $$
\end{lemma}

\begin{remark}
  Note that the sparsity of the signal $\vx$ plays no role in Lemma~\ref{lem:min-l2}; the result holds
  uniformly for all signals $\vx$.
\end{remark}

Let us assume that Lemma~\ref{lem:min-l2} is true for a moment, 
and show how together with Lemma~\ref{lem:L1 RIP} it implies Theorem~\ref{thm:effective sparsity}.

\begin{proof}[Proof of Theorem~\ref{thm:effective sparsity}]
With probability at least $1 - C \exp(-cm)$, the conclusions of both Lemma~\ref{lem:min-l2} and Lemma~\ref{lem:L1 RIP} with $t=1/4$ hold.
Assume this event occurs. 
Consider a signal $\vx$ as in Theorem~\ref{thm:effective sparsity} and the corresponding solution $\hat{\vx}$ of \eqref{eq:convex program}. 
By Lemma~\ref{lem:min-l2}, the latter satisfies
\begin{equation}							\label{eq:twonorm large}
\twonorm{\hat{\vx}} \geq c/\sqrt{\log(2n/m + 2m/n)}.
\end{equation}
Next, consider
$$
\lambda = \frac{1}{m} \sum_{i=1}^m \abs{ \< \va_i, \vx \> } = \frac{1}{m} \|\mA\vx\|_1.
$$
Since by the assumption on $\vx$ we have $\vx/\|\vx\|_2 \in K_{n,s} \cap S^{n-1}$, Lemma~\ref{lem:L1 RIP} with $t=1/4$
implies that 
\begin{equation}							\label{eq:lambda large}
\l \ge \frac{\|\vx\|_2}{2}.
\end{equation}
By definition of $\l$, the vector $\l^{-1}\vx$ is feasible for the program \eqref{eq:convex program}, so the solution $\hat{\vx}$ 
of this program satisfies
$$
\|\hat{\vx}\|_1 \le \|\l^{-1}\vx\|_1 = \l^{-1}\|\vx\|_1.
$$
Putting this together with \eqref{eq:lambda large} and \eqref{eq:twonorm large}, we conclude that
$$
\frac{\onenorm{\hat{\vx}}}{\twonorm{\hat{\vx}}} 
\leq \frac{\onenorm{\vx}}{\lambda \twonorm{\hat{\vx}}} 
\leq \frac{2 \onenorm{\vx}}{\twonorm{\vx} \twonorm{\hat\vx}} 
\leq \frac{\onenorm{\vx}}{\twonorm{\vx}} \cdot C \sqrt{\log(2n/m + 2m/n)}.
$$
This completes the proof of Theorem~\ref{thm:effective sparsity}.
\end{proof}

In the rest of this section we prove Lemma~\ref{lem:min-l2}. The argument based on the observation 
that the set of possible solutions $\hat{\vx}$ of the convex program \eqref{eq:convex program} for all $\vx$ and corresponding $\vy$ 
is finite, and its cardinality can be bounded by $\exp(Cm \log(2n/m+2m/n))$. For each fixed solution $\hat{\vx}$, 
a lower bound on $\|\hat{\vx}\|_2$ will be deduced from Gaussian concentration inequalities, 
and the argument will be finished by taking a union bound over $\hat{\vx}$.

\medskip

It may be convenient to recast the convex minimization program \eqref{eq:convex program} as a {\em linear program} 
by introducing the dummy variables $\vu = (u_1, u_2, \hdots, u_n)$:

\begin{equation} \label{eq:linear-program}
\min \sum_{i=1}^n u_i \quad \text{such that:} \quad
\begin{cases} 
  -u_i \leq x_i' \leq u_i, & i = 1, 2, \hdots, n;\\
  y_i \< \va_i, \vx'\>  \geq 0, & i = 1, 2, \hdots, m;\\
  \frac{1}{m}\sum_{i=1}^m y_i \< \va_i, \vx'\> \geq 1.
\end{cases} 
\end{equation}


The feasible set of the linear program \eqref{eq:linear-program} is a polytope in $\R^{2n}$, and
the linear program attains a solution on a vertex of this polytope.  
Further, since the random Gaussian vectors $\va_i$ are in general position, 
one can check that the solution of the linear program is unique with probability $1$.
Thus, by characterizing these vertices and pointing out the relationship between $u_i$ and $\hat{x}_i$, 
we may reduce the space of possible solutions $\hat{\vx}$.  This is the content of our next lemma. 
Given subsets $T \subset \{1, 2, \hdots, n\}$, $\Omega \subset \{1, 2, \hdots, m\}$, 
we define $\mA_T^\Omega$ to be the submatrix of $\mA$ with columns indexed by $T$ and rows indexed by $\Omega$.  

\begin{lemma}[Vertices of the feasible polytope]  
  With probability $1$, the linear program \eqref{eq:linear-program} attains a solution $(\hat{\vx}, \vu)$ at a point which satisfies the following for some $T \subset \{1, 2, \hdots, n\}$ and $\Omega \subset \{1, 2, \hdots, m\}$:
  \begin{enumerate}
    \item $u_i = \abs{\hat{x}_i}$;
    \item $\supp(\hat{\vx}) = T$;
    \item $\abs{T} = \abs{\Omega} + 1$;
    \item $\mA^\Omega_T \, \hat{\vx}_T = 0$;
    \item $\frac{1}{m} \sum_{i=1}^m \abs{\< \va_i, \hat{\vx}\>} = 1$.
  \end{enumerate}
\end{lemma}

\begin{proof}
Part (1) follows since we are minimizing $\sum u_i$. Part (5) follows since
\[\frac{1}{m} \sum_{i=1}^m y_i \< \va_i, \hat{\vx}\> = \frac{1}{m} \sum_{i=1}^m \abs{\< \va_i, \hat{\vx}\>}\]
combined with the fact that we are implicitly minimizing $\onenorm{\vx'}$. 
Parts (2)--(4) will follow from the fact that \eqref{eq:linear-program} achieves its minimum at a vertex.   
The vertices are precisely the feasible points for which some $d$ of the inequality constraints achieve equality, 
provided $\hat{\vx}$ is the unique solution to those $d$ equalities.  Since $(\hat{\vx}, \vu) \in \R^{2n}$ at least $2n$ of the constraints must be equalities.  We now count equalities based on $T$ and $\Omega$.  

We first consider the constraints $-u_i \leq x_i' \leq u_i,  i = 1, 2, \hdots, n$. 
If $\hat{x}_i = 0$ we have two equalities, $-u_i = \hat{x}_i$ and $u_i = \hat{x}_i$, otherwise, we have one.  This gives $n + \abs{T^c}$ equalities. Part (5) gives one more equality. 
This leaves us with at least $2n - n - \abs{T^c} -1 = \abs{T} - 1$ equalities that must be satisfied out of the equations $y_i \< \va_i, \hat{\vx}\> \geq 0$.  Thus, we may take $|\Omega| = \abs{T} - 1$.
\end{proof}

\begin{proof}[Proof of Lemma \ref{lem:min-l2}]
We may disregard the dummy variables $(u_i)$ and consider that the solution $\hat{\vx} = \vx'$ must satisfy conditions (2)--(5) above 
for some $T$ and $\Omega$.  
We will show that with high probability, any such vector $\vx' \in \R^n$ is lower bounded in the Euclidean norm. 

Let us first fix sets $T$ and $\Omega$, and consider a vector $\vx'$ satisfying (2)--(5). We represent it as
$$
\vx' = \mu \bar{\vx} \quad \text{for some } \mu>0 \text{ and } \|\bar{\vx}\|_2 =1.
$$
Our goal is to lower bound $\mu$. 
By condition (4) above, we have
$\mA^\Omega_T \, \bar{\vx}_T = 0$
which, with probability 1, completely determines the vector $\bar{\vx}$ up to multiplication by $\pm 1$ 
(since $\abs{T} = \abs{\Omega} + 1$ and $\bar{\vx}_{T^c} = 0$).  
Moreover, since $\supp(\bar{\vx}) = \supp(\vx') = T$, we have 
$0 = \mA^\Omega_T \, \bar{\vx}_T = \mA^\Omega \bar{\vx}$, so $\< \va_i, \vx'\> = 0$ for $i \in \Omega$.
Using this with together with condition (5), we obtain
$$
1 = \mu \, \frac{1}{m} \sum_{i=1}^m \abs{\< \va_i, \bar{\vx}\>} 
= \mu \, \frac{1}{m} \sum_{i \notin \Omega} \abs{\< \va_i, \bar{\vx}\>}
$$
and thus
\begin{equation}							\label{eq:norm vs mu}
\twonorm{\vx'} = \mu = \Big( \frac{1}{m} \sum_{i \notin \Omega} \abs{\< \va_i, \bar{\vx}\>} \Big)^{-1}.
\end{equation}
We proceed to upper bound $\frac{1}{m} \sum_{i \notin \Omega} \abs{\< \va_i, \bar{\vx}\>}$.   

Since the random vector $\bar{\vx}$ depends entirely on $\mA^\Omega_T$, it is independent of $\va_i$ for $i \notin \Omega$.  
Thus, by the rotational invariance of the Gaussian distribution, for any fixed vector $\vz$ with unit norm, we have
the following distributional estimates:\footnote{For random variables $X$, $Y$, the distributional inequality $X \stackrel{\text{dist}}{\leq} Y$ means that 
$\P\{X>t\} \le \P\{Y>t\}$ for all $t \in \R$.}
$$
\frac{1}{m} \sum_{i \notin \Omega} \abs{\< \va_i, \bar{\vx}\>} 
\stackrel{\text{dist}}{=} \frac{1}{m} \sum_{i \notin \Omega}  \abs{\< \va_i, \vz\>} 
\stackrel{\text{dist}}{\leq} \frac{1}{m}  \sum_{i=1}^m \abs{\< \va_i, \vz\>}.
$$
The last term is a sum of independent sub-Gaussian random variables, 
and it can be bounded using standard concentration inequalities. 
Specifically, applying Lemma~\ref{lem:concentration} from the Appendix, we obtain 
$$
\Pr{ \frac{1}{m} \sum_{i=1}^m \abs{\< \va_i, \vz \>} > t} \leq C \exp(-c m t^2) \quad \text{for } t \ge 2.
$$
Using \eqref{eq:norm vs mu}, this is equivalent to
$$
\Pr{\twonorm{\vx'} < 1/t} \leq C \exp(-c m t^2) \quad \text{for } t \ge 2.
$$
It is left to upper bound the number of vectors satisfying conditions (2)--(5) and to use the union bound.
Since $\abs{T} = \abs{\Omega} + 1$, the total number of possible choices for $T$ and $\Omega$ (and hence of $\vx'$) is 
$$
\sum_{i=0}^{\min (m, n-1)} {n \choose i + 1} {m \choose i} \leq \exp(C m \log (2n/m + 2m/n)).
$$
Thus, by picking $t = C_0 \sqrt{\log (2n/m + 2m/n)}$ with a sufficiently large absolute constant $C_0$, 
we find that all $\vx'$ uniformly satisfy the required estimate  $\twonorm{\vx'} \geq c/\sqrt{\log (2n/m + 2m/n)}$
with probability at least $1 - \exp(C m \log (2n/m + 2m/n)) \cdot C \exp(-c m t^2) = 1 - C\exp(-c m \log (2n/m + 2m/n))$.
Lemma \ref{lem:min-l2} is proved.
\end{proof}

\section*{Appendix. Uniform concentration inequality.}

In this section we prove concentration inequalities for
$$
\|\mA\vx\|_1 = \sum_{i=1}^m \abs{\< \va_i, \vx\> }.
$$
In the situation where the vector $\vx$ is fixed, we have a sum of independent random variables, 
which can be controlled by standard concentration inequalities:

\begin{lemma}[Concentration]			\label{lem:concentration}
  Let $n, m \in \N$ and $\vx \in \R^n$. Then, for every $t>0$ one has
  $$
  \Pr{ \abs{\frac{1}{m} \sum_{i=1}^m \abs{\< \va_i, \vx\> } - \sqrt{\frac{2}{\pi}} \, \|\vx\|_2 } > t \|\vx\|_2 } 
  \le C \exp(-cmt^2).
  $$
\end{lemma}

\begin{proof}
Without loss of generality we can assume that $\|\vx\|_2 = 1$. 
Then $\< \va_i, \vx\> $ are independent standard normal random variables, so 
$\E \abs{\< \va_i, \vx\> } = \sqrt{2/\pi}$. Therefore $X_i := \abs{\< \va_i, \vx\> } - \sqrt{2/\pi}$
are independent and identically distributed centered random variables. Moreover, $X_i$ are 
sub-gaussian random variable with $\|X_i\|_{\psi_2} \le C$, see \cite[Remark~18]{Vershynin2010}.
An application of Hoeffding-type inequality (see \cite[Proposition~10]{Vershynin2010}) yields
$$
\Pr{ \abs{\frac{1}{m} \sum_{i=1}^m X_i } > t } \le C \exp(-cmt^2).
$$
This completes the proof.
\end{proof}

We will now prove a stronger version of Lemma~\ref{lem:concentration} that is uniform
over all effectively sparse signals $\vx$.

\begin{lemma}[Uniform concentration]		\label{lem:L1 RIP}
  Let $n \in \N$, $t \in [0, \sqrt{2/\pi}]$, and suppose that $m \geq C t^{-4} s \log(2n/s)$. Then 
  $$
  \Pr{\sup_{\vx \in K_{n,s} \cap S^{n-1}} 
    \abs{\frac{1}{m} \sum_{i=1}^m \abs{\< \va_i, \vx\>} - \sqrt{\frac{2}{\pi}}} > t} 
  \leq C \exp(-c m t^2).
  $$
\end{lemma}

\begin{proof}
This is a standard covering argument, although the approximation step requires a little extra care. 
Let $\MM$ be a $t/4$-net of $K_{n,s} \cap S^{n-1}$.  
Since $K_{n,s} \cap S^{n-1} \subseteq K_{n,s}$, we can arrange by Lemma \ref{lem:entropy compressible} that
$$
\abs{\MM} \leq \exp(C t^{-2} s \log(2n/s)).
$$
By definition, for any $\vx \in K_{n,s} \cap S^{n-1}$ one can find $\bar{\vx} \in \MM$ such that 
$\twonorm{\vx - \bar{\vx}} \leq t/4$. So the triangle inequality yields
$$
\abs{\frac{1}{m}\sum_{i=1}^m \abs{\< \va_i, \vx \>} - \sqrt{\frac{2}{\pi}}} 
\leq \abs{\frac{1}{m}\sum_{i=1}^m \abs{\< \va_i, \bar{\vx} \>} - \sqrt{\frac{2}{\pi}}} + \frac{1}{m}\sum_{i=1}^m \abs{\< \va_i, \vx - \bar{\vx} \>}.
$$
Note that $\onenorm{\vx - \bar{\vx}} \le \onenorm{\vx} + \onenorm{\bar{\vx}} \le  2 \sqrt{s}$. 
Together with $\twonorm{\vx - \bar{\vx}} \leq t/4$ this means that 
$$
\vx - \bar{\vx} \in \frac{t}{4} \cdot K_{n, 64s/t^2}.
$$
Consequently, 
\begin{align}		\label{eq:bound-using-cover}
\sup_{\vx \in K_{n,s} \cap S^{n-1}} \abs{\frac{1}{m} \sum_{i=1}^m \abs{\< \va_i, \vx\>} - \sqrt{\frac{2}{\pi}}} 
&\le \sup_{\bar{\vx} \in \MM} \abs{\frac{1}{m}\sum_{i=1}^m \abs{\< \va_i, \bar{\vx} \>} - \sqrt{\frac{2}{\pi}}}  
  +  \frac{t}{4} \cdot \sup_{\vw \in K_{n, 64s/t^2}} \frac{1}{m}\sum_{i=1}^m \abs{\< \va_i, \vw \>} \\
&=: R_1 + \frac{t}{4} \cdot R_2. \nonumber
\end{align}
We bound the terms $R_1$ and $R_2$ separately.
For simplicity of notation, we assume that $64s/t^2$ is an integer, as the non-integer case will have no significant effect on the result. 
 
A bound on $R_1$ follows from the concentration estimate in Lemma~\ref{lem:concentration} and a union bound:  
\begin{equation}		\label{eq:R1}
\Pr{R_1 > t/4} 
\leq \abs{\MM} \cdot C \exp(- c m t^2) 
\leq C \exp(C t^{-2} s \log(2n/s) - c m t^2) 
\leq C \exp(-cm t^2)
\end{equation}
provided that $m \geq C t^{-4} s \log(2n/s)$.

Next, due to Lemma \ref{lem:characterization} and Jensen's inequality, we have
$$
R_2 \le  2 \sup_{\vw \in S_{n, 64s/t^2}}\frac{1}{m} \sum_{i=1}^m \abs{\< \va_i, \vw\>} 
\leq 2 \sup_{\vw \in S_{n,64s/t^2}} \Big( \frac{1}{m}\sum_{i=1}^m \< \va_i, \vw\> ^2 \Big)^{1/2}
=: 2 R_2'.
$$
The quantity $R_2'$ has been well studied in compressed sensing; it is bounded by the restricted isometry 
constant of the matrix $\frac{1}{\sqrt{m}} A$ at sparsity level $64 s/t^2$. 
Probabilistic bounds for the restricted isometry constants of Gaussian matrices are well known,
and have been derived in the earliest compressed sensing works \cite{Candes2006}.  
We use the bound in \cite[Theorem~65]{Vershynin2010} that gives 
\begin{equation}
\Pr{R_2' > 1.5} \leq 2 \exp(-cm) 
\end{equation}
provided that $m \geq C t^{-2} s \log(n/s)$. Thus 
$$
\Pr{R_2 > 3} \leq 2 \exp(-c m).
$$
Combining this and \eqref{eq:R1} we conclude that  
$$
\Pr{R_1 + \frac{t}{4} \cdot R_2 > t} \le C' \exp(-cm t^2)
$$
where we used the assumption that $t \le \sqrt{2/\pi}$.
This and \eqref{eq:bound-using-cover} complete the proof. 
\end{proof}

\bibliographystyle{acm}
\bibliography{pv-1bitcs}

\end{document}